\documentclass[english]{amsart}
\usepackage{mathptmx}
\usepackage[T1]{fontenc}
\usepackage{amsthm}
\usepackage{graphicx}
\usepackage{amssymb}

\providecommand{\tabularnewline}{\\}

\numberwithin{equation}{section}
\numberwithin{figure}{section}
\newtheorem{thm}{Theorem}
  \newtheorem{example}[thm]{Example}
  \newtheorem{defn}[thm]{Definition}
  \newtheorem{lem}[thm]{Lemma}
  \newtheorem{notation}[thm]{Notation}
  \newtheorem{prop}[thm]{Proposition}
  \newtheorem{rem}[thm]{Remark}
  \newtheorem{algorithm}[thm]{Algorithm}
  \newtheorem{cor}[thm]{Corollary}

\makeatother

\begin{document}

\title{Cylindrical Algebraic Decomposition Using Local Projections}

\author{Adam Strzebo\'nski}

\address{Wolfram Research Inc., 100 Trade Centre Drive, Champaign, IL 61820,
U.S.A. }

\email{adams@wolfram.com}

\maketitle

\begin{abstract}
We present an algorithm which computes a cylindrical algebraic decomposition
of a semialgebraic set using projection sets computed for each cell
separately. Such local projection sets can be significantly smaller
than the global projection set used by the Cylindrical Algebraic Decomposition
(CAD) algorithm. This leads to reduction in the number of cells the
algorithm needs to construct. We give an empirical comparison of our
algorithm and the classical CAD algorithm.
\end{abstract}

\section{Introduction}

A semialgebraic set is a subset of $\mathbb{R}^{n}$ which is a solution
set of a system of polynomial equations and inequalities. Computation
with semialgebraic sets is one of the core subjects in computer algebra
and real algebraic geometry. A variety of algorithms have been developed
for real system solving, satisfiability checking, quantifier elimination,
optimization and other basic problems concerning semialgebraic sets
\cite{C,BPR,CJ,CMXY,DSW,GV,HS,LW,R,T,W1}. Every semialgebraic set
can be represented as a finite union of disjoint cells bounded by
graphs of algebraic functions. The Cylindrical Algebraic Decomposition
(CAD) algorithm \cite{C,CJ,S7} can be used to compute a cell decomposition
of any semialgebraic set presented by a quantified system of polynomial
equations and inequalities. An alternative method of computing cell
decompositions is given in \cite{CMXY}. Cell decompositions computed
by the CAD algorithm can be represented directly \cite{S7,S8,B2}
as cylindrical algebraic formulas (CAF; see the next section for a
precise definition). A CAF representation of a semialgebraic set $A$
can be used to decide whether $A$ is nonempty, to find the minimal
and maximal values of the first coordinate of elements of $A$, to
generate an arbitrary element of $A$, to find a graphical representation
of $A$, to compute the volume of $A$, or to compute multidimensional
integrals over $A$ (see \cite{S4}). 

The CAD algorithm takes a system of polynomial equations and inequalities
and constructs a cell decomposition of its solution set. The algorithm
consists of two phases. The projection phase finds a set of polynomials
whose roots are sufficient to describe the cell boundaries. The lifting
phase constructs a cell decomposition, one dimension at a time, subdividing
cells at all roots of the projection polynomials. However, some of
these subdivisions may be unnecessary, either because of the geometry
of the roots or because of the Boolean structure of the input system.
In this paper we propose an algorithm which combines the two phases.
It starts with a sample point and constructs a cell containing the
point on which the input system has a constant truth value. Projection
polynomials used to construct the cell are selected based on the structure
of the system at the sample point. Such a local projection set can
often be much smaller than the global projection set used by the CAD
algorithm. The idea to use such locally valid projections was first
introduced in \cite{JM}, in an algorithm to decide the satisfiability
of systems of real polynomial equations and inequalities. It was also
used in \cite{B4}, in an algorithm to construct a single open cell
from a cylindrical algebraic decomposition. 
\begin{example}
\label{exa:MainExample}Find a cylindrical algebraic decomposition
of the solution set of $S=f_{1}<0\vee(f_{2}\leq0\wedge f_{3}\leq0)$,
where $f_{1}=4x^{2}+y^{2}-4$, $f_{2}=x^{2}+y^{2}-1$, and $f_{3}=16x^{6}-24x^{4}+9x^{2}+4y^{4}-4y^{2}$.

\includegraphics[width=0.65\columnwidth, trim = -50mm 3mm 20mm 5mm, clip]{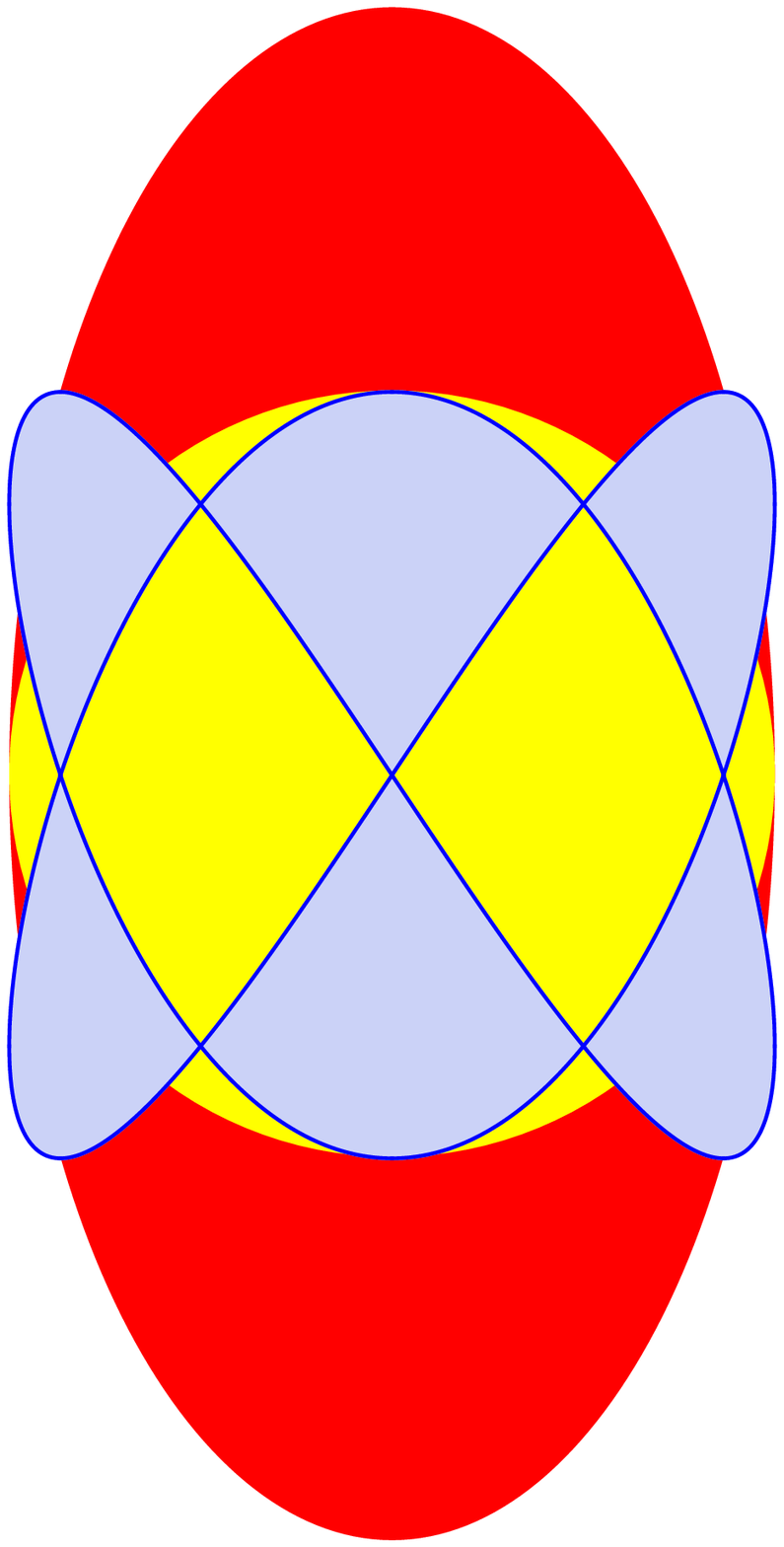}

The solution set of $S$ is equal to the union of the open ellipse
$f_{1}<0$ and the intersection of the closed disk $f_{2}\leq0$ and
the set $f_{3}\leq0$ bounded by a Lissajous curve. As can be seen
in the picture, the set is equal to the open ellipse $f_{1}<0$. The
CAD algorithm uses a projection set consisting of the discriminants
and the pairwise resultants of $f_{1}$, $f_{2}$, and $f_{3}$. It
computes a cell decomposition of the solution set of $S$ by constructing
$357$ cells such that all $f_{1}$, $f_{2}$, and $f_{3}$ have a
constant sign on each cell. Note however, that a cell decomposition
of the solution set of $S$ can be obtained by considering the following
$13$ cells. On each cell only some of $f_{1}$, $f_{2}$, and $f_{3}$
have a constant sign, but those signs are sufficient to determine
the truth value of $S$.
\begin{enumerate}
\item $S$ is $true$ on $-1<x<1\wedge-2\sqrt{1-x^{2}}<y<2\sqrt{1-x^{2}}$
because $f_{1}<0$.
\item $S$ is $false$ on $-1<x<1\wedge y<-2\sqrt{1-x^{2}}$ and on $-1<x<1\wedge y>2\sqrt{1-x^{2}}$
because $f_{1}>0\wedge f_{2}>0$.
\item $S$ is $false$ on $-1<x<1\wedge y=-2\sqrt{1-x^{2}}$ and on $-1<x<1\wedge y=2\sqrt{1-x^{2}}$
because $f_{1}=0\wedge f_{2}>0$.
\item $S$ is $false$ on $x<-1$ and on $x>1$ because $f_{1}>0\wedge f_{2}>0$.
\item $S$ is $false$ on $x=-1\wedge y<0$ and on $x=-1\wedge y>0$ because
$f_{1}>0\wedge f_{2}>0$.
\item $S$ is $false$ on $x=-1\wedge y=0$ because $f_{1}=0\wedge f_{3}>0$.
\item $S$ is $false$ on $x=1\wedge y<0$ and on $x=1\wedge y>0$ because
$f_{1}>0\wedge f_{2}>0$.
\item $S$ is $false$ on $x=1\wedge y=0$ because $f_{1}=0\wedge f_{3}>0$.
\end{enumerate}
Determining the cell bounds for the cell stack $(1)$-$(3)$ requires
computation of roots of $discr_{y}f_{1}$, $discr_{y}f_{2}$, and
$res_{y}(f_{1},f_{2})$ in $x$ and roots of $f_{1}(0,y)$ and $f_{2}(0,y)$
in $y$. Determining the cell bounds for the cells $(4)$ requires
computation of roots of $discr_{y}f_{1}$ and $discr_{y}f_{2}$ in
$x$ and roots of $f_{1}(-2,y)$, $f_{2}(-2,y)$, $f_{1}(2,y)$ and
$f_{2}(2,y)$ in $y$. Determining the cell bounds for the cell stacks
$(5)$-$(6)$ and $(7)$-$(8)$ requires computation of roots of $f_{1}(-1,y)$,
$f_{2}(-1,y)$, $f_{3}(-1,y)$, $f_{1}(1,y)$, $f_{2}(1,y)$ and $f_{3}(1,y)$
in $y$. Polynomial $f_{3}$ is not used to compute any of the projections
and its roots in $y$ are computed only for two values of $x$. The
algorithm we propose in this paper computes a cell decomposition of
the solution set of $S$ by constructing the $13$ cells given in
$(1)$-$(8)$. Details of the computation for this example are given
in Section \ref{sub:Example}.
\end{example}

\begin{example}
Find a cylindrical algebraic decomposition of the solution set of
$S=ax^{4}+bx^{3}+cx^{2}+dx+e\geq0$ in the variable order $(a,b,c,d,e,x)$.

In this example the system is not well-oriented, hence the CAD algorithm
needs to use Hong's projection operator for the first three projections.
However, the additional projection polynomials are necessary only
for the cells on which a McCallum's projection polynomial vanishes
identically. For most cells local projection can be computed using
McCallum's projection operator, and for the few cells on which a McCallum's
projection polynomial vanishes identically local projection needs
to use some, but usually not all, polynomials from Hong's projection
operator. The algorithm LPCAD we propose in this paper computes a
cell decomposition of the solution set of $S$ by constructing $523$
cells in $0.95$ seconds of CPU time. The CAD algorithm did not finish
the computation in $72$ hours. A version of LPCAD using only local
projections based on Hong's projection operator constructs $1375$
cells and takes $2.72$ seconds of CPU time. 
\end{example}

\section{Preliminaries}

A \emph{system of polynomial equations and inequalities} in variables
$x_{1},\ldots,x_{n}$ is a formula\[
S(x_{1},\ldots,x_{n})=\bigvee_{1\leq i\leq l}\bigwedge_{1\leq j\leq m}f_{i,j}(x_{1},\ldots,x_{n})\rho_{i,j}0\]
where $f_{i,j}\in\mathbb{R}[x_{1},\ldots,x_{n}]$, and each $\rho_{i,j}$
is one of $<,\leq,\geq,>,=,$ or $\neq$. 

A subset of $\mathbb{R}^{n}$ is \emph{semialgebraic} if it is a solution
set of a system of polynomial equations and inequalities. 

A \emph{quantified system of real polynomial equations and inequalities}
in free variables $x_{1},\ldots,x_{n}$ and quantified variables $t_{1},\ldots,t_{m}$
is a formula \[
Q_{1}t_{1}\ldots Q_{m}t_{m}S(t_{1},\ldots,t_{m};x_{1},\ldots,x_{n})\]
 Where $Q_{i}$ is $\exists$ or $\forall$, and $S$ is a system
of real polynomial equations and inequalities in $t_{1},\ldots,t_{m},x_{1},\ldots,x_{n}$.

By Tarski's theorem (see \cite{T}), solution sets of quantified systems
of real polynomial equations and inequalities are semialgebraic.
\begin{notation}
For $k\geq1$, let $\overline{a}$ denote a $k$-tuple $(a_{1},\ldots,a_{k})$
of real numbers and let $\overline{x}$ denote a $k$-tuple $(x_{1},\ldots,x_{k})$
of variables.
\end{notation}
Every semialgebraic set can be represented as a finite union of disjoint
\emph{cells} (see \cite{L}), defined recursively as follows.
\begin{enumerate}
\item A cell in $\mathbb{R}$ is a point or an open interval.
\item A cell in $\mathbb{R}^{k+1}$ has one of the two forms\begin{eqnarray*}
 & \{(\overline{a},a_{k+1}):\overline{a}\in C_{k}\wedge a_{k+1}=r(\overline{a})\}\\
 & \{(\overline{a},a_{k+1}):\overline{a}\in C_{k}\wedge r_{1}(\overline{a})<a_{k+1}<r_{2}(\overline{a})\}\end{eqnarray*}
where $C_{k}$ is a cell in $\mathbb{R}^{k}$, $r$ is a continuous
algebraic function, and $r_{1}$ and $r_{2}$ are continuous algebraic
functions, $-\infty$, or $\infty$, and $r_{1}<r_{2}$ on $C_{k}$. 
\end{enumerate}
A finite collection $D$ of cells in $\mathbb{R}^{n}$ is \emph{cylindrically
arranged} if for any $C_{1},C_{2}\in D$ and $k\leq n$ the projections
of $C_{1}$ and $C_{2}$ on $\mathbb{R}^{k}$ are either disjoint
or identical. 

Given a semialgebraic set presented by a quantified system of polynomial
equations and inequalities, the CAD algorithm can be used to decompose
the set into a cylindrically arranged finite collection of cells.
The collection of cells is represented by a cylindrical algebraic
formula (CAF). A CAF describes each cell by giving explicit algebraic
function bounds and the Boolean structure of a CAF reflects the cylindrical
arrangement of cells. Before we give a formal definition of a CAF,
let us first introduce some terminology.

Let $k\geq1$ and let $f=c_{d}y^{d}+\ldots+c_{0}$, where $c_{0},\ldots,c_{d}\in\mathbb{\mathbb{Z}}[\overline{x}]$.
A \emph{real algebraic function} given by the \emph{defining polynomial}
$f$ and a \emph{root number} $p\in\mathbb{N}_{+}$ is the function\begin{equation}
Root_{y,p}f:\mathbb{R}^{k}\ni\overline{a}\longrightarrow Root_{y,p}f(\overline{a})\in\mathbb{R}\label{rootfun}\end{equation}
where $Root_{y,p}f(\overline{a})$ is the $p$-th real root of $f(\overline{a},y)\in\mathbb{R}[y]$.
The function is defined for those values of $\overline{a}$ for which
$f(\overline{a},y)$ has at least $p$ real roots. The real roots
are ordered by the increasing value and counted with multiplicities.
A real algebraic number $Root_{y,p}f\in\mathbb{R}$ given by a \emph{defining
polynomial} $f\in\mathbb{Z}[y]$ and a \emph{root number} $p$ is
the $p$-th real root of $f$. See \cite{S2,S4} for more details
on how algebraic numbers and functions can be implemented in a computer
algebra system.

Let $C$ be a connected subset of $\mathbb{R}^{k}$. $Root_{y,p}f$
is\emph{ regular} on \emph{$C$} if it is continuous on $C$, $c_{d}(\overline{a})\neq0$
for all $\overline{a}\in C$, and there exist\emph{ }$m\in\mathbb{\mathbb{N}}_{+}$
such that for any $\overline{a}\in C$ $Root_{y,p}f(\overline{a})$
is a root of $f(\overline{a},y)$ of multiplicity $m$. 

$f$ is \emph{degree-invariant} on $C$ if there exist\emph{ }$e\in\mathbb{\mathbb{N}}$
such that if $c_{d}(\overline{a})=\ldots=c_{e+1}(\overline{a})=0\wedge c_{e}(\overline{a})\neq0$
for all $\overline{a}\in C$. 

A set $W=\{f_{1},\ldots,f_{m}\}$ of polynomials is \emph{delineable}
on $C$ if all elements of $W$ are degree-invariant on $C$ and for
$1\leq i\leq m$\[
f_{i}^{-1}(0)\cap(C\times\mathbb{R})=\{r_{i,1},\ldots,r_{i,l_{i}}\}\]
where $r_{i,1},\ldots,r_{i,l_{i}}$ are disjoint regular real algebraic
functions and for $i_{1}\neq i_{2}$ $r_{i_{1},j_{1}}$ and $r_{i_{2},j_{2}}$
are either disjoint or equal. Functions $r_{i,j}$ are \emph{root
functions of $f_{i}$ over $C$}.

A set $W=\{f_{1},\ldots,f_{m}\}$ of polynomials is \emph{analytic
delineable} on a connected analytic submanifold $C$ of $\mathbb{R}^{k}$
if $W$ is delineable on $C$ and the root functions of elements of
$W$ over $C$ are analytic.

Let $W$ be delineable on $C$, let $r_{1}<\ldots<r_{l}$ be all root
functions of elements of $W$ over $C$, and let $r_{0}=-\infty$
and $r_{l+1}=\infty$. For $1\leq i\leq l$, the $i$-th \emph{$W$-section
over $C$} is the set\[
\{(\overline{a},a_{k+1}):\overline{a}\in C\wedge a_{k+1}=r_{i}(\overline{a})\}\]
For $1\leq i\leq l+1$, the $i$-th \emph{$W$-sector over $C$} is
the set\[
\{(\overline{a},a_{k+1}):\overline{a}\in C\wedge r_{i-1}(\overline{a})<a_{k+1}<r_{i}(\overline{a})\}\]

A formula $F$ is an \emph{algebraic constraint} with \emph{bounds}
$BDS(F)$ if it is a level-$k$ equational or inequality constraint
with $1\leq k\leq n$ defined as follows.\emph{ }
\begin{enumerate}
\item \emph{A level}-$1$ \emph{equational constraint} has the form $x_{1}=r$,
where $r$ is a real algebraic number, and $BDS(F)=\{r\}$.
\item \emph{A level}-$1$ \emph{inequality constraint} has the form $r_{1}<x_{1}<r_{2}$,
where $r_{1}$ and $r_{2}$ are real algebraic numbers, $-\infty$,
or $\infty$, and $BDS(F)=\{r_{1},r_{2}\}\setminus\{-\infty,\infty\}$. 
\item \emph{A level}-$k+1$ \emph{equational constraint} has the form $x_{k+1}=r(\overline{x})$,
where $r$ is a real algebraic function, and $BDS(F)=\{r\}$.
\item \emph{A level}-$k+1$ \emph{inequality constraint} has the form $r_{1}(\overline{x})<x_{k+1}<r_{2}(\overline{x})$,
where $r_{1}$ and $r_{2}$ are real algebraic functions, $-\infty$,
or $\infty$, and $BDS(F)=\{r_{1},r_{2}\}\setminus\{-\infty,\infty\}$. 
\end{enumerate}
A level-$k+1$ algebraic constraint $F$ is \emph{regular} on a connected
set $C\subseteq\mathbb{R}^{k}$ if all elements of $BDS(F)$ are regular
on $C$ and, if $F$ is an inequality constraint, $r_{1}<r_{2}$ on
$C$.
\begin{defn}
An \emph{atomic cylindrical algebraic formula (CAF)} $F$ in $(x_{1},\ldots,x_{n})$
has the form $F_{1}\wedge\ldots\wedge F_{n}$, where $F_{k}$ is a
level-$k$ algebraic constraint for $1\leq k\leq n$ and $F_{k+1}$
is regular on the solution set of $F_{1}\wedge\ldots\wedge F_{k}$
for $1\leq k<n$. 

\emph{Level-$k$ cylindrical subformulas} are defined recursively
as follows
\begin{enumerate}
\item A level-$n$ cylindrical subformula is a disjunction of level-$n$
algebraic constraints.
\item A level-$k$ cylindrical subformula, with $1\leq k<n$, has the form\[
(F_{1}\wedge G_{1})\vee\ldots\vee(F_{m}\wedge G_{m})\]
where $F_{i}$ are level-$k$ algebraic constraints and $G_{i}$ are
level-$k+1$ cylindrical subformulas.
\end{enumerate}
A \emph{cylindrical algebraic formula (CAF)} is a level-$1$ cylindrical
subformula $F$ such that distributing conjunction over disjunction
in $F$ gives \[
DNF(F)=F_{1}\vee\ldots\vee F_{l}\]
where each $F_{i}$ is an atomic CAF. 
\end{defn}
Given a quantified system of real polynomial equations and inequalities
the CAD algorithm \cite{S7} returns a CAF representation of its solution
set. 
\begin{example}
The following formula $F(x,y,z)$ is a CAF representation of the closed
unit ball.\begin{eqnarray*}
F(x,y,z) & := & x=-1\wedge y=0\wedge z=0\vee\\
 &  & -1<x<1\wedge b_{2}(x,y,z)\vee\\
 &  & x=1\wedge y=0\wedge z=0\\
b_{2}(x,y,z) & := & y=R_{1}(x)\wedge z=0\vee\\
 &  & R_{1}(x)<y<R_{2}(x)\wedge b_{2,2}(x,y,z)\vee\\
 &  & y=R_{2}(x)\wedge z=0\\
b_{2,2}(x,y,z) & := & z=R_{3}(x,y)\vee\\
 &  & R_{3}(x,y)<z<R_{4}(x,y)\vee\\
 &  & z=R_{4}(x,y)\end{eqnarray*}
where \begin{eqnarray*}
R_{1}(x) & = & Root_{y,1}(x^{2}+y^{2})=-\sqrt{1-x^{2}}\\
R_{2}(x) & = & Root_{y,2}(x^{2}+y^{2})=\sqrt{1-x^{2}}\\
R_{3}(x,y) & = & Root_{z,1}(x^{2}+y^{2}+z^{2})=-\sqrt{1-x^{2}-y^{2}}\\
R_{4}(x,y) & = & Root_{z,2}(x^{2}+y^{2}+z^{2})=\sqrt{1-x^{2}-y^{2}}\end{eqnarray*}

\end{example}

\section{CAD construction using local projections}

In this section we describe an algorithm for computing a CAF representation
of the solution set of a system of polynomial equations and inequalities.
The algorithm uses local projections computed separately for each
cell. For simplicity we assume that the system is not quantified.
The algorithm can be extended to quantified systems following the
ideas of \cite{CH}. The algorithm in its version given here does
not take advantage of equational constraints. The use of equational
constraints will be described in the full version of the paper.

The main, recursive, algorithm used for CAD construction is Algorithm
\ref{alg:LPCAD}. Let us sketch the algorithm here, a detailed description
is given later in this section. The input is\emph{ }a system $S(x_{1},\ldots,x_{n})$\emph{
}of polynomial equations and inequalities and a point $\overline{a}=(a_{1},\ldots,a_{k})\in\mathbb{R}^{k}$
with $0\leq k<n$. The algorithm finds a level-$k+1$ cylindrical
subformula $F$ and a set of polynomials \emph{$V\subseteq\mathbb{R}[x_{1},\ldots,x_{k}]$}
such that for any cell $C\subseteq\mathbb{R}^{k}$ containing $\overline{a}$
on which all elements of $V$ have constant signs\emph{ }\[
(x_{1},\ldots,x_{k})\in C\Rightarrow\left(F(x_{1},\ldots,x_{n})\Longleftrightarrow S(x_{1},\ldots,x_{n})\right)\]
The formula $F$ can be interpreted as a description of the solution
set of $S$ as a finite collection of cylindrically arranged cells
in $\mathbb{R}^{n-k}$, parametrized by the values of $(x_{1},\ldots,x_{k})$.
The description is valid locally to $\overline{a}$, where the meaning
of {}``locally'' is determined by $V$. The approach is to find
algebraic constraints \[
G_{1}(\overline{x},x_{k+1}),\ldots,G_{m}(\overline{x},x_{k+1})\]
 and cylindrical subformulas $H_{1},\ldots,H_{m}$ such that the solution
sets of \[
G_{1}(\overline{a},x_{k+1}),\ldots,G_{m}(\overline{a},x_{k+1})\]
form a decomposition of $\mathbb{R}$ and $H_{i}$ describes the solution
set of $S$ locally to $\{\overline{a}\}\times\{x_{k+1}\,:\, G_{i}(\overline{a},x_{k+1})\}$.
To find $G$'s, $H$'s, and $V$ we start with a stack containing
the interval $(-\infty,\infty)$ and until the stack is emptied execute
the following steps. We take an interval $I$ off stack and pick $a_{k+1}\in I$.
If evaluating the $k+1$-variate polynomials in $S$ at $(\overline{a},a_{k+1})$
suffices to establish the truth value of $S$, let $P$ be a set of
$k+1$-variate polynomials in $S$ sufficient to establish the truth
value of $S$ and let $H$ be the truth value. Otherwise, let $H$
and $P$ be, respectively, the formula and the set of polynomials
returned by Algorithm \ref{alg:LPCAD} applied to $S$ and $(\overline{a},a_{k+1})$.
We use projection to compute a set \emph{$W\subseteq\mathbb{R}[x_{1},\ldots,x_{k}]$}
such that $P$ is delineable on any cell containing $\overline{a}$
on which all elements of $W$ have constant signs and we add the elements
of $W$ to $V$. Let $J$ be the interval containing $a_{k+1}$ bounded
by the nearest roots of elements of $P$ and let $G$ be the constraint
on $x_{k+1}$ whose bounds are the corresponding algebraic functions.
Note that if $P$ is delineable on a cell $C$ containing $\overline{a}$
then the elements of $P$ have constant signs on $D=\{(\overline{x},x_{k+1})\,:\,\overline{x}\in C\wedge G(\overline{x},x_{k+1})\}$
and hence $H$ is equivalent to $S$ on $D$. We add $G$ and $H$
to the list of $G$'s, $H$'s, and, if $I\setminus J$ is nonempty,
we add the components of $I\setminus J$ to stack. When the stack
is empty we use projection to compute a set \emph{$W\subseteq\mathbb{R}[x_{1},\ldots,x_{k}]$}
such the set of polynomials whose roots appear as bounds in $G$'s
are delineable on any cell containing $\overline{a}$ on which all
elements of $W$ have constant signs and we add the elements of $W$
to $V$. As required, the formula $F=(G_{1}\wedge H_{1})\vee\ldots\vee(G_{m}\wedge H_{m})$
is equivalent to $S$ on any cell containing $\overline{a}$ on which
all elements of $V$ have constant signs. 

To compute a CAF representation of the solution set of $S$ we call
Algorithm \ref{alg:LPCAD} with $k=0$.
\begin{notation}
We will use the following notations.
\begin{enumerate}
\item For a finite set of polynomials $P$, let $\overline{P}$ denote the
set of irreducible factors of the elements of $P$. 
\item Let $IRR_{k}$ denote the irreducible elements of $\mathbb{R}[x_{1},\ldots,x_{k}]\setminus\mathbb{R}[x_{1},\ldots,x_{k-1}]$.
\item For a set $A\subseteq\mathbb{R}^{n}$ and $k\leq n$ let $\Pi_{k}(A)$
denote the projection of $A$ on $\mathbb{R}^{k}$.
\end{enumerate}
\end{notation}
In this section we assume that all polynomials have coefficients in
a fixed computable subfield $K\subseteq\mathbb{R}$, irreducibility
is understood to be in the ring of polynomials with coefficients in
$K$, irreducible factors are always content-free and chosen in a
canonical way, and finite sets of polynomials are always ordered according
to a fixed linear ordering in the set of all polynomials with coefficients
in $K$. In our implementation $K=\mathbb{Q}$.

Whenever we write \emph{$a=(a_{1},\ldots,a_{k})\in\mathbb{R}^{k}$}
with $k\geq0$ we include the possibility of $a=()$, the only element
of $\mathbb{R}^{0}$.

\subsection{Local projection}
\begin{defn}
Let $P\subseteq\mathbb{R}[x_{1},\ldots,x_{n}]$ be a finite set of
polynomials and let $a=(a_{1},\ldots,a_{n-1})\in\mathbb{R}^{n-1}$,
where $n\geq1$. Let $W=(W_{1},\ldots,W_{n})$ be such that $W_{k}$
is a finite subset of $IRR_{k}$ and $\overline{P}\cap IRR_{k}\subseteq W_{k}$
for $1\leq k\leq n$. $W$ is a \emph{local projection sequence} for
$P$ at $a$ iff, for any $1\leq k<n$ and any cell $C\subseteq\mathbb{R}^{k}$,
if $(a_{1},\ldots,a_{k})\in C$ and all elements of $W_{j}$ for $1\leq j\leq k$
have constant signs on $\Pi_{j}(C)$ then the set of elements of $W_{k+1}$
that are not identically zero on $C\times\mathbb{R}$ is delineable
over $C$. 
\end{defn}
To compute local projections we use the following two projection procedures,
derived, respectively, from McCallum's projection operator \cite{MC1,MC2,B}
and Hong's projection operator \cite{H}.
\begin{algorithm}
\label{alg:LProjMC}(LProjMC)\\
Input:\emph{ }$P=\{p_{1},\ldots,p_{m}\}\subseteq IRR_{k+1}$ \emph{and
}$\overline{a}=(a_{1},\ldots,a_{k})\in\mathbb{R}^{k}$\emph{, where
$k\geq1$.}\textup{}\\
\textup{\emph{Output:}}\textup{ A }\emph{finite set }$Q\subseteq\mathbb{R}[x_{1},\ldots,x_{k}]$.
\begin{enumerate}
\item Put $Q=\emptyset$ and compute $R=\{p\in P\::\:\exists b\in\mathbb{R}\: p(\overline{a},b)=0\}$.
\item For $1\leq i\leq m$ do

\begin{enumerate}
\item Let $p_{i}=q_{d}x_{k+1}^{d}+\ldots+q_{0}$. Put $Q=Q\cup\{q_{d}\}$.
\item If $k>1$ and $q_{d}(\overline{a})=\ldots=q_{0}(\overline{a})=0$
put \[
Q=Q\cup\{q_{d-1},\ldots,q_{0}\}\]
 and continue the loop.
\item If $k>1$, $q_{d}(\overline{a})=0$, and none of $q_{d-1},\ldots,q_{0}$
is a nonzero constant, put $Q=Q\cup\{q_{l}\}$, where $l$ is maximal
such that $q_{l}(\overline{a})\neq0$.
\item Put $Q=Q\cup\{disc_{x_{k+1}}p_{i}\}$.
\item If $p_{i}\in R$ then put \[
Q=Q\cup\{res_{x_{k+1}}(p_{i},p_{j})\::\: i<j\leq m\wedge p_{j}\in R\}\]

\end{enumerate}
\item Return $Q$.
\end{enumerate}
\end{algorithm}
In the next algorithm we use the following notation.
\begin{notation}
Let $f,g\in\mathbb{R}[\overline{x}][x_{k+1}]$, $\overline{a}\in\mathbb{R}^{k}$,
and \[
d=\min(\deg(f),\deg(g))\]
If for some $0\leq l<d$, $psc_{0}(f,g)(\overline{a})=\ldots=psc_{l-1}(f,g)(\overline{a})=0$
and $psc_{l}(f,g)(\overline{a})\neq0$, then $PSC(f,g,\overline{a}):=\{psc_{0}(f,g),\ldots,psc_{l}(f,g)\}$.
Otherwise \[
PSC(f,g,\overline{a}):=\{psc_{0}(f,g),\ldots,psc_{d-1}(f,g)\}\]
\end{notation}
\begin{algorithm}
\label{alg:LProjH}(LProjH)\\
Input:\emph{ }$P=\{p_{1},\ldots,p_{m}\}\subseteq IRR_{k+1}$ \emph{and
}$\overline{a}=(a_{1},\ldots,a_{k})\in\mathbb{R}^{k}$\emph{, where
$k\geq1$.}\textup{}\\
\textup{\emph{Output:}}\textup{ A }\emph{finite set }$Q\subseteq\mathbb{R}[x_{1},\ldots,x_{k}]$. 
\begin{enumerate}
\item Put $Q=\emptyset$ and compute $R=\{p\in P\::\:\exists b\in\mathbb{R}\: p(\overline{a},b)=0\}$.
\item For $1\leq i\leq m$ do

\begin{enumerate}
\item Let $p_{i}=q_{d}x_{k+1}^{d}+\ldots+q_{0}$. Put $Q=Q\cup\{q_{d}\}$
and $r_{i}=p_{i}$.
\item If $q_{d}(\overline{a})=\ldots=q_{0}(\overline{a})=0$ put $Q=Q\cup\{q_{d-1},\ldots,q_{0}\}$
and continue the loop.
\item If $q_{d}(\overline{a})=0$, put $Q=Q\cup\{q_{d-1},\ldots,q_{l}\}$
and $r_{i}=q_{l}x_{x+1}^{l}+\ldots+q_{0}$, where $l$ is maximal
such that $q_{l}(\overline{a})\neq0$.
\item Put $Q=Q\cup PSC(r_{i},\frac{\partial r_{i}}{\partial x_{k+1}},\overline{a})$.
\item If $p_{i}\in R$ then for $i<j\leq m$ if $p_{j}\in R$ put $Q=Q\cup PSC(r_{i},p_{j},\overline{a})$.
\end{enumerate}
\item Return $Q$.
\end{enumerate}
\end{algorithm}
The following algorithm computes a local projection for given $P$
and $a$. 
\begin{algorithm}
\label{alg:LocalProj}(LocalProjection)\\
Input:\emph{ A finite set }$P\subseteq\mathbb{R}[x_{1},\ldots,x_{n}]$
\emph{and }$a=(a_{1},\ldots,a_{n-1})\in\mathbb{R}^{n-1}$\emph{, where
$n\geq1$.}\textup{}\\
\textup{\emph{Output:}}\textup{ A local projection sequence} $W=(W_{1},\ldots,W_{n})$
\emph{for $P$ at $a$.}\textup{ }
\begin{enumerate}
\item Set $wo=true$, $Q=P$, $k=n-1$.
\item While $k\geq1$ do

\begin{enumerate}
\item Let $\overline{a}=(a_{1},\ldots,a_{k})$ and compute $W_{k+1}=\overline{Q}\cap IRR_{k+1}$,
$Q=\overline{Q}\setminus W_{k+1}$.
\item If $wo=true$, $1<k<n-1$, and an element of $W_{k+1}$ is identically
zero at $\overline{a}$, then set $wo=false$, $Q=P$, $k=n-1$ and
continue the loop.
\item If $wo=true$ or $k\leq2$ set $Q=Q\cup LProjMC(W_{k+1},\overline{a})$
else set $Q=Q\cup LProjH(W_{k+1},\overline{a})$.
\item Set $k=k-1$.
\end{enumerate}
\item Set $W_{1}=\overline{Q}\cap IRR_{1}$.
\item Return $W=(W_{1},\ldots,W_{n})$.
\end{enumerate}
\end{algorithm}

\subsection{The CAD construction algorithm}

Let us first introduce an algorithm for evaluation of polynomial systems
at {}``partial'' sample points.
\begin{algorithm}
\label{alg:PEVAL}(PEval)\\
Input:\emph{ A system }$S(x_{1},\ldots,x_{n})$\emph{ of polynomial
equations and inequalities and $\overline{a}=(a_{1},\ldots,a_{k})\in\mathbb{R}^{k}$
with $0\leq k\leq n$.}\textup{}\\
\textup{\emph{Output:}}\textup{ $undecided$ or a pair $(v,P)$,
where $v\in\{true,false\}$, $P=\{p_{1},\ldots,p_{m}\}\subseteq\mathbb{R}[x_{1},\ldots,x_{k}]$,
and for any $b=(b_{1},\ldots,b_{n})\in\mathbb{R}^{n}$ if \[
sign(p_{i}(a_{1},\ldots,a_{k}))=sign(p_{i}(b_{1},\ldots,b_{k}))\]
 for all $1\leq i\leq m$ then the value of $S(b)$ is $v$. }
\begin{enumerate}
\item If $S=false$ or $S=true$ then return $(S,\emptyset)$.
\item If $S=(f\rho0)$, where $\rho$ is one of $<,\leq,\geq,>,=,$ or $\neq$.

\begin{enumerate}
\item If there exists a factor $g$ of $f$ such that $g\in\mathbb{R}[x_{1},\ldots,x_{k}]$
and $g(\overline{a})=0$ then return $(0\rho0,\{g\})$.
\item If $f\in\mathbb{R}[x_{1},\ldots,x_{k}]$ return $(f(\overline{a})\rho0,\{f\})$.
\item Return $undecided$.
\end{enumerate}
\item If $S=T_{1}\wedge\ldots\wedge T_{l}$

\begin{enumerate}
\item For $1\leq i\leq l$ compute $e_{i}=PEval(T_{i},\overline{a})$.
\item If for some $i$ $e_{i}=(false,P_{i})$ then return $(false,P_{i})$.
\item If for all $i$ $e_{i}=(true,P_{i})$ then return $(true,P_{1}\cup\ldots\cup P_{l})$.
\item Return $undecided$.
\end{enumerate}
\item If $S=T_{1}\vee\ldots\vee T_{l}$

\begin{enumerate}
\item For $1\leq i\leq l$ compute $e_{i}=PEval(T_{i},\overline{a})$.
\item If for some $i$ $e_{i}=(true,P_{i})$ then return $(true,P_{i})$.
\item If for all $i$ $e_{i}=(false,P_{i})$ then return $(false,P_{1}\cup\ldots\cup P_{l})$.
\item Return $undecided$.
\end{enumerate}
\end{enumerate}
\end{algorithm}
We can now present a recursive algorithm computing cylindrical algebraic
decomposition using local projections.
\begin{algorithm}
\label{alg:LPCAD}(LPCAD)\\
Input:\emph{ A system $S(x_{1},\ldots,x_{n})$} \emph{of polynomial
equations and inequalities and $\overline{a}=(a_{1},\ldots,a_{k})\in\mathbb{R}^{k}$
with $0\leq k<n$.}\textup{}\\
\textup{\emph{Output:}}\textup{ }\emph{A pair $(F,V)$, where $F$
is a level-$k+1$ cylindrical subformula, }\textup{$V=(V_{1},\ldots,V_{k})$}\emph{,
}\textup{$V_{j}\subseteq\mathbb{R}[x_{1},\ldots,x_{j}]$}\emph{ for
$1\leq j\leq k$, and for any cell $C\subseteq\mathbb{R}^{k}$ if
$\overline{a}\in C$ and for $1\leq j\leq k$ all elements of $V_{j}$
have constant signs on $\Pi_{j}(C)$ then}\textup{\emph{ }}\textup{\[
(x_{1},\ldots,x_{k})\in C\Rightarrow\left(F(x_{1},\ldots,x_{n})\Longleftrightarrow S(x_{1},\ldots,x_{n})\right)\]
}
\begin{enumerate}
\item Compute a disjunctive normal form $S_{DNF}$ and a conjunctive normal
form $S_{CNF}$ of $S$.
\item Set $stack=\{(-\infty,-\infty,<,\infty,\infty,<)\}$ and $A=Q=V_{1}=\ldots=V_{k}=\emptyset$.
\item While $stack\neq\emptyset$ do

\begin{enumerate}
\item Remove a tuple $(u_{1},r_{1},\rho_{1},u_{2},r_{2},\rho_{2})$ from
$stack$. $r_{1},r_{2}$ are algebraic functions of $x_{1},\ldots x_{k}$,
$-\infty$, or $\infty$, $u_{1}=r_{1}(\overline{a})$, $u_{2}=r_{2}(\overline{a})$,
$\rho_{1},\rho_{2}\in\{<,\leq\}$, and the tuple represents the interval
$u_{1}\rho_{1}x_{k+1}\rho_{2}u_{2}$,
\item If $u_{1}=u_{2}$ set $a_{k+1}=u_{1}$ and set $R=\{f\}$, where $r_{1}=Root_{x_{k+1,p}}f$,
else pick a rational number $u_{1}<a_{k+1}<u_{2}$ and set $R=\emptyset$.
Set $\overline{b}=(\overline{a},a_{k+1})$.
\item Compute $e_{CNF}=PEval(S_{CNF},\overline{b})$. If $e_{CNF}=(false,P)$
then set $H=false$ and $W=LocalProjection(P\cup R,\overline{a})$,
and go to $(f)$.
\item Compute $e_{DNF}=PEval(S_{DNF},\overline{b})$. If $e_{DNF}=(true,P)$
then set $H=true$ and $W=LocalProjection(P\cup R,\overline{a})$,
and go to $(f)$.
\item Compute $(H,U)=LPCAD(S,\overline{b})$. For $1\leq j\leq k$ set $V_{j}=V_{j}\cup U_{j}$.
Compute $W=LocalProjection(U_{k+1}\cup R,\overline{a})$.
\item For $1\leq j\leq k$ set $V_{j}=V_{j}\cup W_{j}$.
\item If $u_{1}=u_{2}$ then set $G=(x_{k+1}=r_{1})$ and go to $(n)$.
\item Find $s_{1}$and $s_{2}$ such that 

\begin{enumerate}
\item $s_{1}=Root_{x_{k+1},p_{1}}f_{1}$ and $f_{1}\in W_{k+1}$ or $s_{1}=f_{1}\equiv-\infty$, 
\item $s_{2}=Root_{x_{k+1},p_{2}}f_{2}$ and $f_{2}\in W_{k+1}$ or $s_{2}=f_{2}\equiv\infty$, 
\item $v_{1}=s_{1}(\overline{a})$ and $v_{2}=s_{2}(\overline{a})$, 
\item either $v_{1}=v_{2}=a_{k+1}$ or $v_{1}<a_{k+1}<v_{2}$ and there
are no roots of elements of $W_{k+1}$ in $(v_{1},v_{2})$. 
\end{enumerate}
\item Set $Q=Q\cup(\{f_{1},f_{2}\}\setminus\{-\infty,\infty\})$.
\item If $v_{1}=v_{2}$ then set $G=(x_{k+1}=s_{1})$, add \[
(v_{1},s_{1},<,u_{2},r_{2},\rho_{2})\]
 and \[
(u_{1},r_{1},\rho_{1},v_{1},s_{1},<)\]
 to $stack$, and go to $(n)$.
\item If $u_{2}<v_{2}$ then set $t_{2}=r_{2}$ and $\sigma_{2}=\rho_{2}$.
Else set $t_{2}=s_{2}$ and $\sigma_{2}=<$, and if $u_{2}>v_{2}$
or $\rho_{2}=\leq$ add \[
(v_{2},s_{2},\leq,u_{2},r_{2},\rho_{2})\]
 to $stack$. 
\item If $v_{1}<u_{1}$ then set $t_{1}=r_{1}$ and $\sigma_{1}=\rho_{1}$.
Else set $t_{1}=s_{1}$ and $\sigma_{1}=<$, and if $v_{1}>u_{1}$
or $\rho_{1}=\leq$ add \[
(u_{1},r_{1},\rho_{1},v_{1},s_{1},\leq)\]
 to $stack$. 
\item Set $G=(t_{1}\sigma_{1}x_{k+1}\sigma_{2}t_{2})$.
\item Set $A=A\cup\{(a_{k+1},G\wedge H)\}$
\end{enumerate}
\item Sort A by increasing values of the first element, obtaining $\{(c_{1},H_{1}),\ldots,(c_{m},H_{m})\}$.
Set $F=H_{1}\vee\ldots\vee H_{m}$.
\item Compute $W=LocalProjection(Q,\overline{a})$.
\item For $1\leq j\leq k$ set $V_{j}=V_{j}\cup W_{j}$.
\item Return\emph{ $(F,V)$.}
\end{enumerate}
\end{algorithm}
\begin{cor}
$LPCAD(S(x_{1},\ldots,x_{n}),())$ returns \[
(F(x_{1},\ldots,x_{n}),())\]
where $F(x_{1},\ldots,x_{n})$ is a cylindrical algebraic formula
equivalent to $S(x_{1},\ldots,x_{n})$.
\end{cor}
The formula returned by Algorithm \ref{alg:LPCAD} may involve weak
inequalities, but it can be easily converted to the CAF format by
replacing weak inequalities with disjunctions of equations and strict
inequalities.

\subsection{Proofs}

To prove correctness of Algorithm \ref{alg:LocalProj} we use the
following lemmata.
\begin{lem}
\label{lem:LProjMC}Let\emph{ $k\geq1$}\textup{\emph{,}} $P\subseteq IRR_{k+1}$\emph{,
}$\overline{a}=(a_{1},\ldots,a_{k})\in\mathbb{R}^{k}$\emph{, }\textup{\emph{and
$Q=LProjMC(P,\overline{a})$. If $D$ is}} a connected analytic submanifold
of $\mathbb{R}^{k}$ such that $\overline{a}\in D$ and all elements
of $Q$ are order-invariant in $D$ then the set $P^{*}$ of all elements
of $P$ that are not identically zero on $D\times\mathbb{R}$ is analytic
delineable over $D$ and the elements of $P^{*}$ are order-invariant
in each $P^{*}$-section over $D$. \end{lem}
\begin{proof}
Suppose that $f\in P^{*}$. Step $(2a)$ of Algorithm \ref{alg:LProjMC}
guarantees that $f$ has a sign-invariant leading coefficient in $D$.
$f$ does not vanish identically at any point in $D$ (for $k>1$
it is ensured by step $(2c)$; for $k=1$ it follows from irreducibility
of $f$). By Theorem 3.1 of \cite{B}, $f$ is degree-invariant on
$D$. Since $disc_{x_{k+1}}(f)\in Q$, by Theorem 2 of \cite{MC2},
$\{f\}$ is analytic delineable over $D$ and is order-invariant in
each $\{f\}$-section over $D$. Suppose that $g\in P^{*}$ and $g\neq f$.
If either $f(\bar{a},x_{k+1})$ or $g(\bar{a},x_{k+1})$ has no real
roots then $\{f,g\}$ is delineable on $D$. Otherwise $res_{x_{k+1}}(f,g)\in Q$
and hence, by Theorem 2 of \cite{MC2}, $\{f,g\}$ is analytic delineable
over $D$. Therefore, $P^{*}$ is analytic delineable over $D$ and
the elements of $P^{*}$ are order-invariant in each $P^{*}$-section
over $D$. \end{proof}
\begin{lem}
\label{lem:LProjH}Let\emph{ $k\geq1$}\textup{\emph{,}} $P\subseteq IRR_{k+1}$\emph{,
}$\overline{a}=(a_{1},\ldots,a_{k})\in\mathbb{R}^{k}$\emph{, }\textup{\emph{and
$Q=LProjH(P,\overline{a})$. If $D$ is}} a connected subset of $\mathbb{R}^{k}$
such that $\overline{a}\in D$ and all elements of $Q$ are sign-invariant
in $D$ then the set $P^{*}$ of all elements of $P$ that are not
identically zero on $D\times\mathbb{R}$ is delineable over $D$. \end{lem}
\begin{proof}
Suppose that $f=q_{d}x_{k+1}^{d}+\ldots+q_{0}\in P^{*}$. Let $l$
be maximal such that $q_{l}(\overline{a})\neq0$, and let $f_{red}=q_{l}x_{k+1}^{l}+\ldots+q_{0}$.
Steps $(2a)$ and $(2c)$ of Algorithm \ref{alg:LProjH} guarantee
that $f=f_{red}$ in $D\times\mathbb{R}$. By step $(2d)$ and Theorems
1-3 of \cite{C}, $\{f_{red}\}$ is delineable over $D$, and hence
$\{f\}$ is delineable over $D$. Suppose that $g\in P^{*}$ and $g\neq f$.
If either $f(\bar{a},x_{k+1})$ or $g(\bar{a},x_{k+1})$ has no real
roots then $\{f,g\}$ is delineable on $D$. Otherwise without loss
of generality we may assume that due to step $(2e)$ $Q$ contains
all factors of $PSC(f_{red},g,\overline{a})$. By Lemma 1 of \cite{H}
and Theorem 2 of \cite{C}, the degree of $\gcd(f(\bar{b},x_{k+1}),g(\bar{b},x_{k+1}))$
is constant for $\bar{b}\in D$. Since $f$ and $g$ are degree-invariant
in $D$, by Lemma 12 of \cite{S9}, $\{f,g\}$ is delineable over
$D$. Therefore $P^{*}$ is delineable over $D$. \end{proof}
\begin{prop}
\label{pro:LocalProj}Algorithm \ref{alg:LocalProj} terminates and
returns a local projection sequence for $P$ at $a$.\textup{\emph{ }}\end{prop}
\begin{proof}
To show that the algorithm terminates note that the body of the loop
in step $(2)$ is executed at most $2n-2$ times.

Let $W=(W_{1},\ldots,W_{n})$ be the returned sequence. Steps $(2a)$
and $(3)$ ensure that $W_{k}$ is a finite subset of $IRR_{k}$ and
$\overline{P}\cap IRR_{k}\subseteq W_{k}$ for $1\leq k\leq n$. We
will recursively construct a cell $D\subseteq\mathbb{R}^{n-1}$ such
that $D_{k}=\Pi_{k}(D)$ is the maximal connected set containing $\Pi_{k}(a)$
such that all elements of $W_{j}$ for $1\leq j\leq k$ have constant
signs on $\Pi_{j}(D_{k})$. Moreover, for $1\leq k<n$, the set $W_{k+1}^{*}$
of elements of $W_{k+1}$ that are not identically zero on $D_{k}\times\mathbb{R}$
is delineable over $D_{k}$. This is sufficient to prove that $W$
is a local projection sequence for $P$ at $a$, because for any cell
$C\subseteq\mathbb{R}^{k}$ if $(a_{1},\ldots,a_{k})\in C$ and all
elements of $W_{j}$ for $1\leq j\leq k$ have constant signs on $\Pi_{j}(C)$
then $C\subseteq D_{k}$, by maximality of $D_{k}$.

We will consider two cases depending on the value of $wo$ when the
algorithm terminated. Suppose first that when the algorithm terminated
$wo$ was $true$. In this case we will additionally prove that for
$1\leq k<n$ $D_{k}$ is an analytic submanifold of $\mathbb{R}^{k}$,
all elements of $W_{k}$ are order-invariant in $D_{k}$, and if $k<n-1$
then none of the elements of $W_{k+1}$ vanishes identically at any
point in $D_{k}$, $W_{k+1}$ is analytic delineable on $D_{k}$,
and the elements of $W_{k+1}$ are order-invariant in each $W_{k+1}$-section
over $D_{k}$. If $a_{1}$ is a root of an element of $W_{1}$ let
$D_{1}=\{a_{1}\}$ else let $D_{1}=(r_{1},s_{1})$, where $r_{1}$
and $s_{1}$ are roots of elements of $W_{1}$, $-\infty$, or $\infty$,
$r_{1}<a_{1}<s_{1}$, and there are no roots of $W_{1}$ in $(r_{1},s_{1})$.
$D_{1}$ is a connected analytic submanifold of $\mathbb{R}^{1}$
and all elements of $W_{1}$ are order-invariant in $D_{1}$. Since
the elements of $W_{2}$ are irreducible, none of the elements of
$W_{2}$ vanishes identically at any point in $D_{1}$. Since all
irreducible factors of elements of $LProjMC(W_{2},\Pi_{1}(a))$ belong
to $W_{1}$, by Lemma \ref{lem:LProjMC}, $W_{2}$ is analytic delineable
over $D_{1}$ and the elements of $W_{2}$ are order-invariant in
each $W_{2}$-section over $D_{1}$. Suppose that, for some $1<k<n-1$,
we have constructed $D_{k-1}$ satisfying the required conditions.
The conditions imply that $W_{k}$ is analytic delineable on $D_{k-1}$.
Let $D_{k}$ be the $W_{k}$-section or $W_{k}$-sector over $D_{k-1}$
which contains $\Pi_{k}(a)$. $D_{k}$ is an analytic submanifold
of $\mathbb{R}^{k}$. The elements of $W_{k}$ are order-invariant
in $D_{k}$, because they are order-invariant in each $W_{k}$-section
over $D_{k-1}$ and nonzero in each $W_{k}$-sector over $D_{k-1}$.
Since all irreducible factors of elements of $LProjMC(W_{k+1},\Pi_{k}(a))$
belong to $W_{1}\cup\ldots\cup W_{k}$ , by Lemma \ref{lem:LProjMC},
$W_{k+1}^{*}$ is analytic delineable over $D_{k}$ and the elements
of $W_{k+1}^{*}$ are order-invariant in each $W_{k+1}^{*}$-section
over $D_{k}$. Step $(2b)$ guarantees that if $k<n-1$ then $W_{k+1}^{*}=W_{k+1}$.

Suppose now that when the algorithm terminated $wo$ was $false$.
Let $D_{1}$ be as in the first part of the proof. As before, $W_{2}$
is analytic delineable over $D_{1}$ and the elements of $W_{2}$
are order-invariant in each $W_{2}$-section over $D_{1}$. Let $D_{2}$
be the $W_{2}$-section or $W_{2}$-sector over $D_{1}$ which contains
$(a_{1},a_{2})$. $D_{2}$ is an analytic submanifold of $\mathbb{R}^{2}$.
The elements of $W_{2}$ are order-invariant in $D_{2}$, because
they are order-invariant in each $W_{2}$-section over $D_{1}$ and
nonzero in each $W_{2}$-sector over $D_{1}$. Since all irreducible
factors of elements of $LProjMC(W_{3},\Pi_{2}(a))$ belong to $W_{1}\cup W_{2}$,
by Lemma \ref{lem:LProjMC}, $W_{3}^{*}$ is analytic delineable over
$D_{2}$. Suppose that, for some $2<k<n-1$, we have constructed $D_{k-1}$
satisfying the required conditions. The conditions on $D_{k-1}$ imply
that $W_{k}^{*}$ is delineable on $D_{k-1}$. Let $D_{k}$ be the
$W_{k}^{*}$-section or $W_{k}^{*}$-sector over $D_{k-1}$ which
contains $\Pi_{k}(a)$. All elements of $W_{k}$ are sign-invariant
in $D_{k}$. Since all irreducible factors of elements of $LProjH(W_{k+1},\Pi_{k}(a))$
belong to $W_{1}\cup\ldots\cup W_{k}$ , by Lemma \ref{lem:LProjH},
$W_{k+1}^{*}$ is delineable over $D_{k}$.

Since for $1\leq k<n$, $D_{k}$ is the $W_{k}^{*}$-section or $W_{k}^{*}$-sector
over $D_{k-1}$ which contains $\Pi_{k}(a)$, $D_{k}$ is the maximal
connected set containing $\Pi_{k}(a)$ such that all elements of $W_{j}$
for $1\leq j\leq k$ have constant signs on $D_{j}$. 
\end{proof}
Correctness and termination of Algorithm \ref{alg:PEVAL} is obvious.
\begin{prop}
Algorithm \ref{alg:LPCAD} terminates and the returned pair $(F,V)$
satisfies the required conditions.\end{prop}
\begin{proof}
Let $P_{S}$ be the set of all polynomials that appear in $S$ and
let $W_{H}=(W_{H,1},\ldots,W_{H,n})$ be the Hong's projection sequence
\cite{H} for $P_{S}$ (the variant of given in Proposition 7 of \cite{S9}).
Suppose that $\overline{P}\subseteq W_{H,1}\cup\ldots\cup W_{H,k+1}$
and $\overline{a}\in\mathbb{R}^{k}$, where $k<n$. Let $(W_{1},\ldots,W_{k+1})=LocalProjection(P,\overline{a})$.
Since we assume that finite sets of polynomials are consistently ordered
according to a fixed linear order in the set of all polynomials, $W_{i}\subseteq W_{H,i}$
for $1\leq i\leq k+1$. Hence all polynomials that appear during execution
of $LPCAD$ are elements of $W_{H,1}\cup\ldots\cup W_{H,n}$. In particular,
$r_{1}$ and $r_{2}$ that appear in the elements of $stack$ are
roots of elements of $W_{H,k+1}$, $-\infty$, or $\infty$. Therefore,
the number of possible elements of $stack$ is finite, and hence the
loop in step $(3)$ terminates. Recursive calls to $TDCAD$ increment
$k$. When $k=n-1$ then either step $(3c)$ yields $H=false$ or
step $(3d)$ yields $H=true$, and hence step $(3e)$ containing the
recursive call to $LPCAD$ is never executed. Therefore the value
of $k$ is bounded by $n-1$, and hence the recursion terminates.

Let $(F,V)$ be the pair returned by $LPCAD$ and suppose that $C\subseteq\mathbb{R}^{k}$
is a cell such that $\overline{a}\in C$ and for $1\leq j\leq k$
all elements of $V_{j}$ have constant signs on $\Pi_{j}(C)$. We
need to show that\emph{ }\[
(x_{1},\ldots,x_{k})\in C\Rightarrow\left(F(x_{1},\ldots,x_{n})\Longleftrightarrow S(x_{1},\ldots,x_{n})\right)\]
Let $c=(c_{1},\ldots,c_{n})\in\mathbb{R}^{n}$ and $\bar{c}=(c_{1},\ldots,c_{k})\in C$.
We need to show that $F(c)=S(c)$. Let $W=LocalProjection(Q,\overline{a})$,
as computed in step $(5)$. All elements of $W_{j}$ have constant
signs on on $\Pi_{j}(C)$, for $1\leq j\leq k$. Since none of the
elements of $Q$ vanishes identically at $\overline{a}$, $Q$ is
delineable over $C$. Hence the $Q$-sections and the $Q$-sectors
over $C$ form a partition of $C\times\mathbb{R}$. 

For a tuple $\theta=(u_{1},r_{1},\rho_{1},u_{2},r_{2},\rho_{2})$
that appears on $stack$ in any iteration of the loop in step $(3)$
put\[
Z_{1}(\theta)=\{(\bar{x},x_{k+1})\in\mathbb{R}^{k+1}\::\:\bar{x}\in C\wedge r_{1}\rho_{1}x_{k+1}\rho_{2}r_{2}\}\]
For each $\alpha=(a_{k+1},G\wedge H)\in A$ put \[
Z_{2}(\alpha)=\{(\bar{x},x_{k+1})\in\mathbb{R}^{k+1}\::\:\bar{x}\in C\wedge G(\bar{x},x_{k+1})\}\]
Note that each $Z_{1}(\theta)$ and $Z_{2}(\alpha)$ is a union of
$Q$-sections and $Q$-sectors over $C$. Put $\Omega_{1}=\{Z_{1}(\theta)\::\,\theta\in stack\}$
and $\Omega_{2}=\{Z_{2}(\alpha)\::\,\alpha\in A\}$. We will show
that in each instance of the loop in step $(3)$ $\Omega_{1}\cup\Omega_{2}$
is a partition of $C\times\mathbb{R}$. In the first instance of the
loop in step $(3)$ $\Omega_{1}=\{C\times\mathbb{R}\}$ and $\Omega_{2}=\emptyset$,
and hence $\Omega_{1}\cup\Omega_{2}$ is a partition of $C\times\mathbb{R}$.
We will show that this property is preserved in each instance of the
loop. In each instance a tuple $\theta=(u_{1},r_{1},\rho_{1},u_{2},r_{2},\rho_{2})$
is removed from $stack$ and $\alpha=(a_{k+1},G\wedge H)$ is added
to $A$. If $u_{1}=u_{2}$ in step $(3g)$ then $Z_{2}(\alpha)=Z_{1}(\theta)$
and the property is preserved. If $v_{1}=v_{2}$ in step $(3j)$ then
$G=(x_{k+1}=s_{1})$ and tuples $\theta_{2}=(v_{1},s_{1},<,u_{2},r_{2},\rho_{2})$
and $\theta_{1}=(u_{1},r_{1},\rho_{1},v_{1},s_{1},<)$ are added to
$stack$. Since $\{Z_{1}(\theta_{1}),Z_{2}(\alpha),Z_{1}(\theta_{2})\}$
is a partition of $Z_{1}(\theta)$, the property is preserved. Otherwise
steps $(3k)$-$(3m)$ are executed. If in step $(3k)$ $u_{2}>v_{2}$
or $u_{2}=v_{2}$ and $\rho_{2}=\leq$ then put $Z_{1,2}=Z_{1}(\theta_{2})$,
where $\theta_{2}=(v_{2},s_{2},\leq,u_{2},r_{2},\rho_{2})$ is the
tuple added to $stack$, else put $Z_{1,2}=\emptyset$. If in step
$(3l)$ $v_{1}>u_{1}$ or $v_{1}=u_{1}$ and $\rho_{1}=\leq$ then
put $Z_{1,1}=Z_{1}(\theta_{1})$, where $\theta_{1}=(u_{1},r_{1},\rho_{1},v_{1},s_{1},\leq)$
is the tuple added to $stack$, else put $Z_{1,1}=\emptyset$. Since
$\{Z_{1,1},Z_{2}(\alpha),Z_{1,2}\}$ is a partition of $Z_{1}(\theta)$,
the property is preserved. 

After the loop in step $(3)$ is finished $stack$ is empty, $\Omega_{1}=\emptyset$,
and hence $\Omega_{2}$ is a partition of $C\times\mathbb{R}$. Let
$\alpha=(a_{k+1},G\wedge H)\in A$ be such that $(\bar{c},c_{k+1})\in Z_{2}(\alpha)$.
Let us analyze the instance of the loop in step $(3)$ which resulted
in adding $\alpha$ to $A$. Let $D=Z_{2}(\alpha)$. 

Suppose first that $H=false$ or $H=true$ was found in step $(3c)$
or $(3d)$. Let $W=LocalProjection(P\cup R,\overline{a})$, as computed
in step $(3c)$ or $(3d)$. For $1\leq j\leq k$, $W_{j}\subseteq V_{j}$,
and hence all elements of $W_{j}$ have constant signs on on $\Pi_{j}(D)$.
Therefore the set $W_{k+1}^{*}$ of elements of $W_{k+1}$ that are
not identically zero on $C\times\mathbb{R}$ is delineable over $C$.
By definition of $G$, $D$ is a $W_{k+1}^{*}$-section or a $W_{k+1}^{*}$-sector
over $C$. Hence all elements of $W_{k+1}$ have constant signs on
$D$. In particular, all elements of $P$ have constant signs on $D$,
and so $S(c)=H=F(c)$. 

Now suppose that $(H,U)=LPCAD(S,\overline{b})$ was computed in step
$(3e)$. Let \[
W=LocalProjection(U_{k+1}\cup R,\overline{a})\]
For $1\leq j\leq k$, $W_{j}\subseteq V_{j}$, and hence all elements
of $W_{j}$ have constant signs on on $\Pi_{j}(D)$. As before, $W_{k+1}^{*}$
is delineable over $C$, $D$ is a $W_{k+1}^{*}$-section or a $W_{k+1}^{*}$-sector
over $C$, and all elements of $W_{k+1}$ have constant signs on $D$.
In particular, all elements of $U_{k+1}$ have constant signs on $D$.
Since for $1\leq j\leq k$ $U_{j}\subseteq V_{j}$, all elements of
$U_{j}$ have constant signs on on $\Pi_{j}(D)$. Hence\[
(x_{1},\ldots,x_{k},x_{k+1})\in D\Rightarrow\left(H(x_{1},\ldots,x_{n})\Longleftrightarrow S(x_{1},\ldots,x_{n})\right)\]
and so $F(c)=H(c)=S(c)$.
\end{proof}

\subsection{Implementation remarks}
\begin{rem}
\label{rem:Zdim}The following somewhat technical improvements have
been observed to improve practical performance of Algorithm \ref{alg:LPCAD}.
\begin{enumerate}
\item In step $(2c)$ of Algorithm \ref{alg:LProjMC} in $q_{l}$ may be
chosen arbitrarily as long as $q_{l}(\overline{a})\neq0$, hence an
implementation may choose the simplest $q_{l}$.
\item If in a recursive call to $LPCAD(S,(a_{1},\ldots,a_{k}))$ the initial
coordinates $(a_{1},\ldots,a_{m})$ correspond to single-point intervals,
that is $u_{1}=u_{2}$ in step $(3b)$ of the currently evaluated
iteration of loop $(3)$ in all parent computations of \[
LPCAD(S,(a_{1},\ldots,a_{j}))\]
 for $1\leq j\leq m$, then $LocalProjection(P,(a_{1},\ldots,a_{k}))$
does not need to compute the last $m$ levels of projection. Instead
it can return $W=(W_{1},\ldots,W_{n})$ with $W_{1}=\ldots=W_{m}=\emptyset$.
\item Computations involved in finding projections are repeated multiple
times. A practical implementation needs to make extensive use of caching.
\end{enumerate}
\end{rem}

\subsection{\label{sub:Example}Example}

In this section we apply $LPCAD$ to solve the problem stated in Example
\ref{exa:MainExample}.

In step $(1)$ of $LPCAD(S,())$ we compute $S_{CNF}=(f_{1}<0\vee f_{2}\leq0)\wedge(f_{1}<0\vee f_{3}\leq0)$
and $S_{DNF}=f_{1}<0\vee(f_{2}\leq0\wedge f_{3}\leq0)$. In the first
iteration of loop $(3)$ we remove a tuple representing $-\infty<x<\infty$
from $stack$ and pick $a_{1}=0$. The calls to $PEval$ in steps
$(3c)$ and $(3d)$ yield $undecided$. Step $(3e)$ makes a recursive
call to $LPCAD(S,(0))$.

In the first iteration of loop $(3)$ in $LPCAD(S,(0))$ we remove
a tuple representing $-\infty<y<\infty$ from $stack$ and pick $a_{2}=0$.
$PEval(S_{CNF},(0,0))$ in step $(3c)$ yields $(true,\{f_{1},f_{2},f_{3}\})$.
We continue on to step $(3d)$ where $PEval(S_{DNF},(0,0))$ yields
$(true,\{f_{1}\})$. We set $H=true$ and compute \[
W=LocalProjection(\{f_{1}\},(0))=(W_{1},\{f_{1}\})\]
where $W_{1}=\{x-1,x+1\}$ is the set of factors of $discr_{y}f_{1}=16(x^{2}-1)$.
We go to step $(3f)$ and set $V_{1}=V_{1}\cup W_{1}=\{x-1,x+1\}$.
In step $(3h)$ we find $s_{1}=Root_{y,1}f_{1}=-2\sqrt{1-x^{2}}$,
$s_{2}=Root_{y,2}f_{1}=2\sqrt{1-x^{2}}$, $v_{1}=-2$, and $v_{2}=2$.
In step $(3i)$ we set $Q=Q\cup\{f_{1}\}=\{f_{1}\}$. In steps $(3k)$
and $(3l)$ we add tuples representing $2\leq y<\infty$ and $-\infty<y\leq-2$
to $stack$. In step $(3n)$ we obtain $A=\{(0,-2\sqrt{1-x^{2}}<y<2\sqrt{1-x^{2}})\}$.

In the second iteration of loop $(3)$ in $LPCAD(S,(0))$ we remove
a tuple representing $-\infty<y\leq-2$ from $stack$ and pick $a_{2}=-4$.
$PEval(S_{CNF},(0,-4))$ in step $(3c)$ yields $(false,\{f_{1},f_{2}\})$.
We set $H=false$ and compute \[
W=LocalProjection(\{f_{1},f_{2}\},(0))=(W_{1},\{f_{1},f_{2}\})\]
where $W_{1}=\{x-1,x+1\}$ is the set of factors of $discr_{y}f_{1}=16(x^{2}-1)$,
$discr_{y}f_{2}=4(x^{2}-1)$, and $res_{y}(f_{1},f_{2})=9(x^{2}-1)^{2}$.
We go to step $(3f)$ and set $V_{1}=V_{1}\cup W_{1}=\{x-1,x+1\}$.
In step $(3h)$ we find $s_{1}=v_{1}=-\infty$, $s_{2}=Root_{y,1}f_{1}=-2\sqrt{1-x^{2}}$,
and $v_{2}=-2$. In step $(3i)$ we set $Q=Q\cup\{f_{1}\}=\{f_{1}\}$.
In step $(3k)$ we add a tuple representing $y=-2$ to $stack$. In
step $(3n)$ we obtain $A=\{(0,-2\sqrt{1-x^{2}}<y<2\sqrt{1-x^{2}}),(-4,false)\}$.

In the third iteration of loop $(3)$ in $LPCAD(S,(0))$ we remove
a tuple representing $y=-2$ from $stack$ and set $a_{2}=-2$. $PEval(S_{CNF},(0,-2))$
in step $(3c)$ yields \[
(false,\{f_{1},f_{2}\})\]
We set $H=false$ and compute \[
W=LocalProjection(\{f_{1},f_{2}\},(0))=(W_{1},\{f_{1},f_{2}\})\]
where $W_{1}=\{x-1,x+1\}$. We go to step $(3f)$ and set $V_{1}=V_{1}\cup W_{1}=\{x-1,x+1\}$.
In step $(3g)$ we set $G=(y=-2\sqrt{1-x^{2}})$. In step $(3n)$
we obtain $A=\{(0,-2\sqrt{1-x^{2}}<y<2\sqrt{1-x^{2}}),(-4,false),(-2,false)\}$.

The remaining two iterations of loop $(3)$ look very similar to the
last two. In step $(4)$ we obtain $F=-2\sqrt{1-x^{2}}<y<2\sqrt{1-x^{2}}$.
In step $(5)$ we compute \[
W=LocalProjection(\{f_{1}\},(0))=(\{x-1,x+1\},\{f_{1}\})\]
 and in step $(6)$ we set $V_{1}=V_{1}\cup W_{1}=\{x-1,x+1\}$. The
returned value is $(-2\sqrt{1-x^{2}}<y<2\sqrt{1-x^{2}},(\{x-1,x+1\}))$.

In step $(3e)$ of $LPCAD(S,())$ we obtain $H=-2\sqrt{1-x^{2}}<y<2\sqrt{1-x^{2}}$
and $U=(\{x-1,x+1\})$. \[
LocalProjection(\{x-1,x+1\},())\]
 yields $(\{x-1,x+1\})$. In step $(3h)$ we find $s_{1}=Root_{x,1}(x+1)=-1$,
$s_{2}=Root_{x,1}(x-1)=1$, $v_{1}=-1$, and $v_{2}=1$. In steps
$(3k)$ and $(3l)$ we add tuples representing $1\leq x<\infty$ and
$-\infty<x\leq-1$ to $stack$. In step $(3n)$ we obtain $A=\{(0,-1<x<1\wedge-2\sqrt{1-x^{2}}<y<2\sqrt{1-x^{2}})\}$.

In the second iteration of loop $(3)$ in $LPCAD(S,())$ we remove
a tuple representing $-\infty<x\leq-1$ from $stack$ and pick $a_{1}=-2$.
The calls to $PEval$ in steps $(3c)$ and $(3d)$ yield $undecided$.
Step $(3e)$ makes a recursive call to $LPCAD(S,(-2))$.

In the first iteration of loop $(3)$ in $LPCAD(S,(-2))$ we remove
a tuple representing $-\infty<y<\infty$ from $stack$ and pick $a_{2}=0$.
$PEval(S_{CNF},(-2,0))$ in step $(3c)$ yields \[
(false,\{f_{1},f_{2}\})\]
We set $H=false$ and compute \[
W=LocalProjection(\{f_{1},f_{2}\},(-2))=(W_{1},\{f_{1},f_{2}\})\]
where $W_{1}=\{x-1,x+1\}$ is the set of factors of $discr_{y}f_{1}$
and $discr_{y}f_{2}$ ($res_{y}(f_{1},f_{2})$ is not a part of the
projection because $f_{1}(-2,y)$ and $f_{2}(-2,y)$ have no real
roots). We go to step $(3f)$ and set $V_{1}=V_{1}\cup W_{1}=\{x-1,x+1\}$.
In step $(3h)$ we find $s_{1}=v_{1}=-\infty$ and $s_{2}=v_{2}=\infty$.
In step $(3i)$ $Q$ remains empty. In step $(3n)$ we obtain $A=\{(0,false)\}$.
The loop ends after one iteration and the returned value is $(false,(\{x-1,x+1\}))$.

In step $(3e)$ of $LPCAD(S,())$ we obtain $H=false$ and $U=(\{x-1,x+1\})$.
\[
LocalProjection(\{x-1,x+1\},())\]
yields $(\{x-1,x+1\})$. In step $(3h)$ we find$s_{1}=v_{1}=-\infty$,
$s_{2}=Root_{x,1}(x+1)=-1$, and $v_{2}=-1$. In step $(3k)$ we add
a tuple representing $x=-1$ to $stack$. In step $(3n)$ we obtain
$A=\{(0,-1<x<1\wedge-2\sqrt{1-x^{2}}<y<2\sqrt{1-x^{2}}),(-2,false)\}$.

In the third iteration of loop $(3)$ in $LPCAD(S,())$ we remove
a tuple representing $x=-1$ from $stack$ and pick $a_{1}=-2$. The
calls to $PEval$ in steps $(3c)$ and $(3d)$ yield $undecided$.
Step $(3e)$ makes a recursive call to $LPCAD(S,(-1))$.

In the first iteration of loop $(3)$ in $LPCAD(S,(-1))$ we remove
a tuple representing $-\infty<y<\infty$ from $stack$ and pick $a_{2}=0$.
$PEval(S_{CNF},(-1,0))$ in step $(3c)$ yields \[
(false,\{f_{1},f_{3}\})\]
We set $H=false$ and compute \[
W=LocalProjection(\{f_{1},f_{3}\},(-1))=(W_{1},\{f_{1},f_{2}\})\]
where, by Remark \ref{rem:Zdim}, we can take $W_{1}=\emptyset$.
We go to step $(3f)$ and the set $V_{1}$ remains empty. In step
$(3h)$ we find $s_{1}=s_{2}=Root_{y,1}f_{1}$ and $v_{1}=v_{2}=0$.
In step $(3i)$ we set $Q=Q\cup\{f_{1}\}=\{f_{1}\}$. In step $(3j)$
we add tuples representing $0<x<\infty$ and $-\infty<x<0$ to $stack$.
In step $(3n)$ we obtain $A=\{(0,false)\}$. 

In the second iteration of loop $(3)$ in $LPCAD(S,(-1))$ we remove
a tuple representing $-\infty<y<0$ from $stack$ and pick $a_{2}=-1$.
$PEval(S_{CNF},(-1,-1))$ in step $(3c)$ yields $(false,\{f_{1},f_{2}\})$.
We set $H=false$ and compute $W=LocalProjection(\{f_{1},f_{2}\},(-1))=(W_{1},\{f_{1},f_{2}\})$,
where, by Remark \ref{rem:Zdim}, we can take $W_{1}=\emptyset$.
We go to step $(3f)$ and the set $V_{1}$ remains empty. In step
$(3h)$ we find $s_{1}=v_{1}=-\infty$, $s_{2}=Root_{y,1}f_{1}$ and
$v_{2}=0$. In step $(3i)$ we set $Q=Q\cup\{f_{1}\}=\{f_{1}\}$.
In step $(3n)$ we obtain $A=\{(0,false),(-1,false)\}$. 

The remaining iteration of loop $(3)$ look very similar to the last
one. In step $(4)$ we obtain $F=false$. In step $(5)$ we compute
\[
W=LocalProjection(\{f_{1}\},(-1))=(\emptyset,\{f_{1}\})\]
by Remark \ref{rem:Zdim}. The returned value is $(false,(\emptyset))$.

In step $(3e)$ of $LPCAD(S,())$ we obtain $H=false$ and $U=(\emptyset)$.
$LocalProjection(\emptyset,())$ yields $(\emptyset)$. In step $(3g)$
we set $G=(x=-1)$. In step $(3n)$ we obtain $A=\{(0,-1<x<1\wedge-2\sqrt{1-x^{2}}<y<2\sqrt{1-x^{2}}),(-2,false),(-1,false)\}$.

The remaining two iterations of loop $(3)$ look very similar to the
last two. In step $(4)$ we obtain $F=-1<x<1\wedge-2\sqrt{1-x^{2}}<y<2\sqrt{1-x^{2}}$
and the returned value is $(-1<x<1\wedge-2\sqrt{1-x^{2}}<y<2\sqrt{1-x^{2}},())$.

\section{Empirical Results}

Algorithm \ref{alg:LPCAD} ($LPCAD$) and the cylindrical algebraic
decomposition ($CAD$) algorithm have been implemented in C, as a
part of the kernel of \emph{Mathematica}. The experiments have been
conducted on a Linux server with a $32$-core $2.4$ GHz Intel Xeon
processor and $378$ GB of RAM available for all processes. The reported
CPU time is a total from all cores used. Since we do not describe
the use of equational constraints in the current paper, we have selected
examples that do not involve equations.

\subsection{Benchmark examples}

We compare the performance of $LPCAD$ and $CAD$ for the following
three problems and for the $7$ examples from Wilson's benchmark set
\cite{W2} (version 4) that do not contain equations. 
\begin{example}
\label{exa:Two-quadratics}(Two quadratics) Find a cylindrical algebraic
decomposition of the solution set of $ax^{2}+bx+c\geq0\wedge dx^{2}+ex+f\geq0$
with the variables ordered $(a,b,c,d,e,f,x)$.
\end{example}

\begin{example}
\label{exa:Ellipse in a square}(Ellipse in a square) Find conditions
for ellipse $\frac{(x-c)^{2}}{a}+\frac{(y-d)^{2}}{b}<1$ to be contained
in the square $-1<x<1\wedge-1<y<1$. We compute a cylindrical algebraic
decomposition of the solution set of\begin{eqnarray*}
 & \forall x,y\in\mathbb{R}\; a>0\wedge b>0\wedge b(x-c)^{2}+a(y-d)^{2}<ab\Rightarrow\\
 & -1<x<1\wedge-1<y<1\end{eqnarray*}
with the free variables ordered $(a,b,c,d)$.
\end{example}

\begin{example}
\label{exa:Distance to three squares}(Distance to three squares)
Find the distance of a point on the parabola shown in the picture
to the union of three squares.

\includegraphics[width=0.65\columnwidth, trim = -50mm 3mm 10mm 5mm, clip]{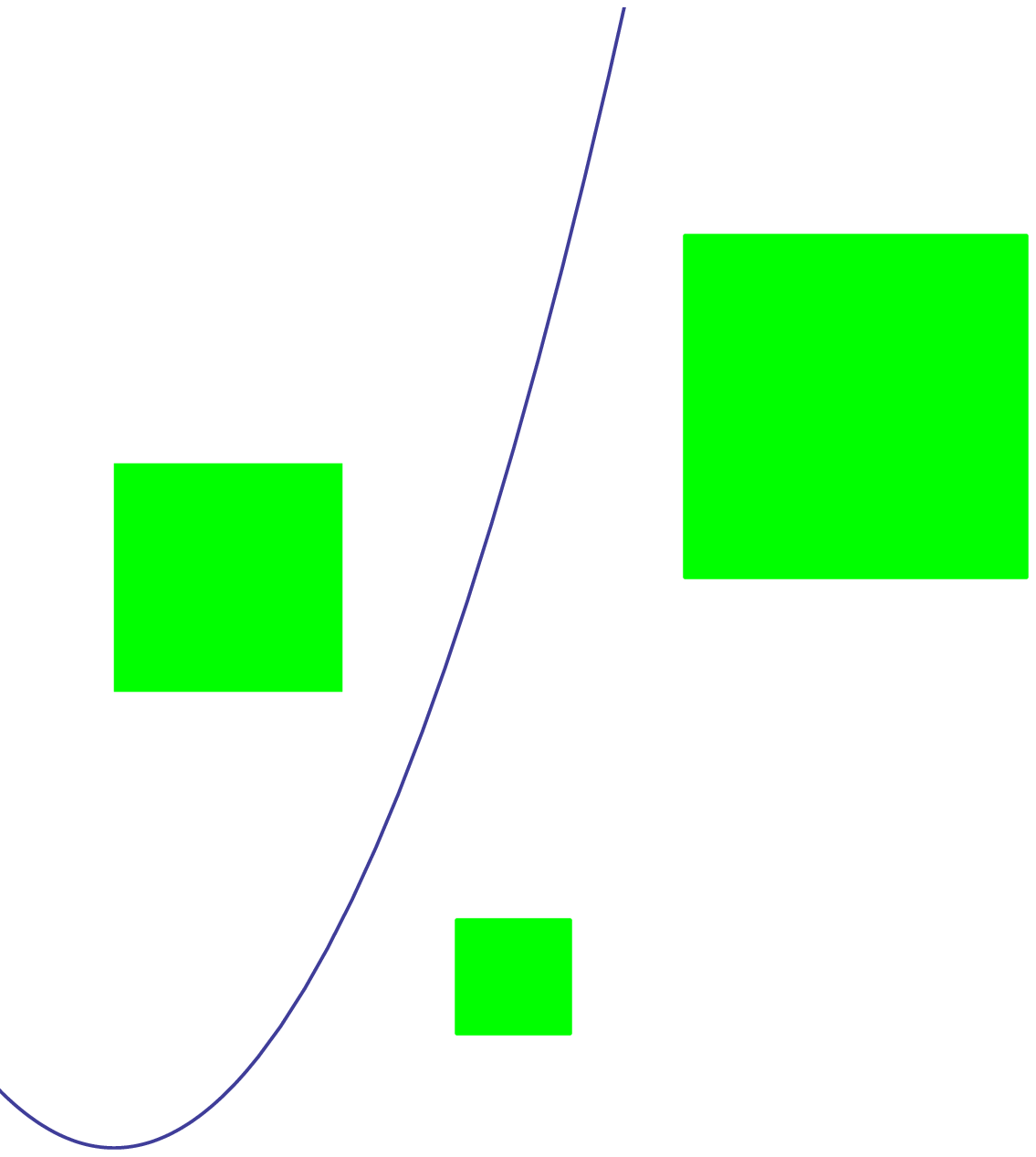}

We compute a cylindrical algebraic decomposition of the solution set
of\begin{eqnarray*}
 & \exists x,y\in\mathbb{R}\;(x-a)^{2}+(y-a^{2}+2)^{2}\leq d\wedge\\
 & (0\leq x\leq1\wedge0\leq y\leq1\vee\\
 & \frac{3}{2}\leq x\leq2\wedge-\frac{3}{2}\leq y\leq-1\vee\\
 & \frac{5}{2}\leq x\leq4\wedge\frac{1}{2}\leq y\leq2)\end{eqnarray*}
with the free variables ordered $(a,d)$.
\end{example}
Results of experiments are given in Table \ref{tab:Benchmark-examples}.
Examples from \cite{W2} are marked with W and the original number.
The columns marked Time give the CPU time, in seconds, used by each
algorithm. The columns marked Cells give the number of cells constructed
by each algorithm. The column marked WO tells whether the system is
well-oriented.

\begin{table}

\caption{\label{tab:Benchmark-examples}Benchmark examples}
\begin{tabular}{|c|c|c|c|c|c|}
\hline 
Example & \multicolumn{2}{c|}{Time} & \multicolumn{2}{c|}{Cells} & WO\tabularnewline
\cline{2-5} 
 & $CAD$ & $LPCAD$ & $CAD$ & $LPCAD$ & \tabularnewline
\hline 
\ref{exa:Two-quadratics} & $97.7$ & $2.61$ & $324137$ & $3971$ & N\tabularnewline
\hline 
\ref{exa:Ellipse in a square} & $>100000$ & $38.1$ & $?$ & $67535$ & N\tabularnewline
\hline 
\ref{exa:Distance to three squares} & $2402$ & $44.9$ & $13105366$ & $71411$ & Y\tabularnewline
\hline 
W 2.3 & $0.063$ & $0.088$ & $91$ & $84$ & Y\tabularnewline
\hline 
W 2.8 & $0.015$ & $0.015$ & $15$ & $15$ & Y\tabularnewline
\hline 
W 2.9 & $0.047$ & $0.011$ & $59$ & $19$ & Y\tabularnewline
\hline 
W 2.10 & $0.135$ & $0.197$ & $779$ & $647$ & Y\tabularnewline
\hline 
W 2.11 & $0.045$ & $0.007$ & $463$ & $31$ & N\tabularnewline
\hline 
W 2.16 & $0.076$ & $0.025$ & $644$ & $4$ & Y\tabularnewline
\hline 
W 6.5 & $2.10$ & $1.58$ & $11279$ & $2536$ & Y\tabularnewline
\hline
\end{tabular}

\end{table}

\subsection{Randomly generated examples}

For this experiment we used randomly generated systems with $5$,
$6$, and $7$ variables, $25$ systems with each number of variables.
The systems had the form $f<0$ or $f\leq0$, with a quadratic polynomial
$f$ with $6$ to $15$ terms and $10$-bit integer coefficients.
We selected systems for which at least one of the algorithms finished
in $1000$ seconds. Results of experiments are given in Table \ref{tab:Random}.
The columns marked Time give the ratio of $CAD$ timing divided by
$LPCAD$ timing. The columns marked Cells give the ratio of the numbers
of cells constructed by $CAD$ and by $LPCAD$. The ratios are computed
for the examples for which both algorithms finished in $1000$ seconds.
The columns marked Mean give geometric means. The column marked TO
gives the number of examples for which $CAD$ did not finish in $1000$
seconds. $LPCAD$ finished in $1000$ seconds for all examples. The
column marked WO gives the number of systems that were well-oriented.

\begin{table}
\caption{\label{tab:Random}Randomly generated examples}
\begin{tabular}{|c|c|c|c|c|c|c|c|c|}
\hline 
Var & \multicolumn{3}{c|}{Time} & \multicolumn{3}{c|}{Cells} & \multicolumn{1}{c|}{TO} & WO\tabularnewline
No. & \multicolumn{3}{c|}{$CAD/LPCAD$} & \multicolumn{3}{c|}{$CAD/LPCAD$} &  & \tabularnewline
\cline{2-7} 
 & Mean & \multicolumn{1}{c|}{Min} & \multicolumn{1}{c|}{Max} & Mean & Min & Max &  & \tabularnewline
\hline 
$5$ & $1.64$ & $0.50$ & $11.1$ & $2.55$ & $0.75$ & $17.3$ & $8$ & $4$\tabularnewline
\hline 
$6$ & $3.82$ & $0.80$ & $55.7$ & $6.14$ & $1$ & $98.4$ & $1$ & $10$\tabularnewline
\hline 
$7$ & $26.9$ & $5.10$ & $257$ & $43.2$ & $6.74$ & $408$ & $3$ & $0$\tabularnewline
\hline
\end{tabular} 
\end{table}

\subsection{Conclusions}

Experiments suggest that for systems that are not well-oriented LPCAD
performs better than CAD. For well oriented-systems LPCAD usually
construct less cells than CAD, but this does not necessarily translate
to a faster timing, due to overhead from re-constructing projection
for every cell. However, for some of the well-oriented systems, for
instance Example \ref{exa:Distance to three squares}, LPCAD is significantly
faster than CAD, due to its ability to exploit the Boolean structure
of the problem. Unfortunately we do not have a precise characterisation
of such problems. Nevertheless LPCAD may be useful for well-oriented
problems that prove hard for the CAD algorithm or may be tried in
parallel with the CAD algorithm.

\bibliographystyle{abbrv}

\end{document}